\documentclass[10pt]{article}

\usepackage{PRIMEarxiv}
\usepackage[symbol]{footmisc}

\usepackage{tensor}

\usepackage{setspace}
\usepackage{makeidx}
\usepackage[dvipsnames]{xcolor}
\definecolor{bluecite}{HTML}{0875b7}
\usepackage{hyperref}\hypersetup{
    colorlinks = true,
    linkcolor = bluecite,
    anchorcolor = bluecite,
    citecolor = bluecite,
    filecolor = bluecite,
    urlcolor = bluecite
   }
\usepackage[sorting=none,citestyle=numeric-comp]{biblatex}
\usepackage[T1]{fontenc}    
\usepackage{url}            
\usepackage{newpxtext}
\usepackage{booktabs}       
\usepackage{mathpazo}

\usepackage{amsfonts}       
\usepackage{nicefrac}       
\usepackage{microtype}      
\usepackage{lipsum}
\RequirePackage[]{tocloft}
\usepackage{amsthm}
\usepackage{amsmath}
\usepackage{fancyhdr}       
\usepackage{graphicx}  
\usepackage{amssymb}
\usepackage{tikz-cd}
\newcounter{definition}
 \newtheorem{Definition}[definition]{Definition}

  \newcounter{remark}
 \newtheorem{Remark}[remark]{Remark}

  \newcounter{example}
 \newtheorem{Example}[example]{Example}

 \newcounter{theorem}
 \newtheorem{Theorem}[theorem]{Theorem}
 
 \newcounter{lemma}
 \newtheorem{Lemma}[lemma]{Lemma}
 
 \newcounter{corollary}
 \newtheorem{Corollary}[corollary]{Corollary}
 
 \newcounter{proposition}
 \newtheorem{Proposition}[proposition]{Proposition}

\title{Enlargement of symmetry groups in physics:\\
a practitioner's guide

}

\author{
  Lehel Csillag \thanks{Author to whom any correspondence should be addressed}\\
Department of Mathematics and Computer Science, Transilvania University\\
Maniu Street 50, Brașov 500091, Romania and\\
Department of Physics,
Babeș-Bolyai University  \\
Kogălniceanu Street, Cluj Napoca 400084, Romania\\
\texttt{lehel.csillag@unitbv.ro, lehel.csillag@ubbcluj.ro, lehel@csillag.ro}
\And
Julio M. Hoff da Silva \\
 Departamento de F\'isica,
Universidade
Estadual Paulista, UNESP \\
Av. Dr. Ariberto Pereira da Cunha, 333, Guaratinguet\'a, SP,
Brazil\\
\texttt{julio.hoff@unesp.br} \\
  \And
Tudor Pătuleanu \\
Department of Physics, West University of Timișoara \\
Bd. Vasile Pârvan 4, Timișoara 300223, Romania \\
\texttt{tudor.patuleanu@e-uvt.ro} \\
}

\begin{document}
\maketitle

\begin{abstract}Wigner's classification has led to the insight that projective unitary representations play a prominent role in quantum mechanics. The physics literature often states that the theory of projective unitary representations can be reduced to the theory of ordinary unitary representations by enlarging the group of physical symmetries. Nevertheless, the enlargement process is not always described explicitly: it is unclear in which cases the enlargement has to be done to the universal cover, a central extension, or to a central extension of the universal cover. On the other hand, in the mathematical literature, projective unitary representations were extensively studied, and famous theorems such as the theorems of Bargmann and Cassinelli have been achieved. The present article bridges the two: we provide a precise, step-by-step guide on describing projective unitary representations as unitary representations of the enlarged group. Particular focus is paid to the difference between algebraic and topological obstructions. To build the bridge mentioned above, we present a detailed review of the difference between group cohomology and Lie group cohomology. This culminates in classifying Lie group central extensions by smooth cocycles around the identity. Finally, the take-away message is a hands-on algorithm that takes the symmetry group of a given quantum theory as input and provides the enlarged group as output. This algorithm is applied to several cases of physical interest. We also briefly outline a generalization of Bargmann's theory to time-dependent phases using Hilbert bundles.
\end{abstract}


\tableofcontents

\section{Introduction}

From a modern perspective, a complete particle characterization is given once one constructs an irreducible representation of the Poincaré group. In this context, for instance, the very concept of spin arises as a quantum label without which the spacetime relativistic symmetries could not be entirely represented in the Hilbert space \cite{wig}. Although the construction presented in \cite{wig} provides an exhaustive classification of the particle content in Quantum Field Theory, it is worthwhile to mention that the concept of spin as intrinsic angular momentum dates back to Uhlenbeck and Goudsmit in $1925$.

A decade after the Stern-Gerlach experiment was conducted in $1922$, spin was introduced ad-hoc in non-relativistic quantum mechanics {by Pauli \cite{pauli1927}} and it was believed that the correct gyromagnetic ratio can only be explained using a relativistic theory \cite{levine2014}. While it is true that the interaction between spin and orbital angular momentum {can only be described in a relativistic theory, as Dirac has shown in \cite{Dirac1928}}, spin already appears in the non-relativistic setting. The key theoretical observation to which the concept of spin could be traced back appears in the works of Wigner \cite{wig} and Bargmann \cite{barg}, who pointed out that in quantum mechanics, instead of using unitary representations of a group $G$, one has to use unitary projective (or ray) representations. This is because any two quantum states differing only by an unimodular phase lead to the same probabilities, as the probabilities are the measurable quantities in quantum mechanics. Hence, (pure) quantum states are rays in a projective Hilbert space instead of elements of a Hilbert space. Symmetries of quantum systems have, therefore, to be implemented as bijections on the projective Hilbert space, which preserve transition probabilities. Wigner's theorem \cite{lomont1963,emch1963,bargmann1964,simon2008,gyory2004,sharma1990} shows that every such map can be implemented on the Hilbert space as a unitary or anti-unitary operator up to a phase. 

Phases appearing in projective representations are well known to depend on the group element being represented and not dependent, in general, on the physical state upon which the representation takes place. {However, it has been argued that time (or even spacetime) phase dependence may connect projective representations with precise experiments \cite{polo,cow,kchu,nesv}. These experiments have in common that the gravitational field acts on the phase of a wave function}. The relation between time or spacetime phase dependence and gauge freedom is widely understood. It is, therefore, only natural to generalize projective representation theory, {to allow for the inclusion of a  spacetime dependence in the phase}. This generalization was performed in \cite{voa}, and we describe its main aspects in the Appendix \ref{appendixA} section. As we shall see, it brings interesting aspects of generalization and, at the same time, fits the main line of presentation for this work since it falls within the algorithm we construct along the paper. 

Being motivated by foundational aspects of non-relativistic quantum mechanics, {Lévy-}Leblond \cite{levy1963,levy1967} has extensively studied the unitary ray representations of the Galilei group. However, it is important to note that he missed one representation, as pointed out recently by O'Farrill and his collaborators \cite{figueroaofarrill2024}.  In the seminal paper on non-relativistic wave-equations \cite{levy1967}, {Lévy-}Leblond also discovered a Galilei-invariant wave equation, through which the correct gyromagnetic ratio could be explained from first principles within a non-relativistic theory. Earlier attempts at formulating Galilei-invariant equations are also present in the literature \cite{galindo1961,eberlein1962}, although they are not completely satisfactory. The equation presented in \cite{galindo1961} is not invariant under the full Galilean group, only under one of its subgroups, which the authors call the static Galilei group. For a detailed review of the difference between the Pauli and the Lévy-Leblond equations, see \cite{wilkes2020}.

Outside the Galilean group, several other possible kinematic groups exist for non-relativistic theories, such as the Carroll or Bargmann groups. A complete classification can be found in \cite{bacry1968}. Although this classification is quite old, there is a recent growing interest in studying non-relativistic theories motivated by non-relativistic holography \cite{zaanen2015,christensen2014}, soft theorems \cite{weinberg1965} and Bondi–Metzner–Sachs (BMS) symmetries \cite{duval2014}, non-relativistic string theories \cite{gomis2001,cardona2016} and Carrollian physics \cite{levyleblond1965,sengupta1966}, which allows for understanding the symmetries of null hypersurfaces \cite{donnay2019}. For a detailed review of non-Lorentzian physics, consult \cite{bergshoeff2023}.

The underlying mathematical concept, which treats both the Carroll and Bargmann groups on the same footing, is that of a central extension \cite{figueroa2023}. Historically, central extensions from a physical perspective were first studied in quantum mechanics in the context of projective representations and geometric quantization \cite{tuynman1987}, but also modern treatments of symmetry actions in quantum mechanics are present in the literature \cite{cassinelli2004}. However, {the latter} might be quite abstract for a physicist. From a physicist's perspective, projective representations are briefly described in Weinberg's modern textbook {on quantum field theory} \cite{weinberg2005}. Therein, at the level of rigor of a physicist, it is pointed out that for any {symmetry group} (Lie group) $G$, one can enlarge {$G$ to $G_{enl}$}, {such that the projective unitary representations of $G$ are replaced with the unitary representations of $G_{enl}$.} Weinberg also introduces central charges, which essentially are generators of the centrally extended Lie algebra associated with $G_{enl}$.

Motivated by the recently growing interest in non-relativistic kinematical symmetry groups in both classical and quantum physics and the interaction between anomalies and central extensions \cite{freed2023}, the goal of the present review article is to give a concrete algorithm of finding $G_{enl}$ as present in Weinberg's textbook, using mathematical properties of the Lie group $G$ in question, based on previous results already presented in the literature, making quite explicit all the assumptions. In particular, we would like to be as precise as possible and build a bridge between the mathematical and physical literature on this topic, opening the possibility for future productive collaborations between mathematicians and physicists. Each mathematical theorem presented will be exemplified through direct applications in physics.

The paper is structured as follows. We start with an introduction and a short recapitulation of states and observables in quantum mechanics to place both mathematicians and physicists on the same foot and level of rigor. Symmetries á la Wigner are presented, then the central notions of (projective) unitary representations are defined, considering strict topological conditions. In section \ref{section2}, we turn our attention to lifting projective representations of a Lie group $G$ to ordinary representations of an enlarged group, as it turns out that projective representations can't always be lifted to ordinary representations of the same group. We classify the obstructions from an algebraic and a topological perspective. As projective representations are in one-to-one correspondence with $U(1)$ group extensions and cohomology theory, we give a detailed exposition on Lie group cohomology, which is only briefly presented in the literature from a completely different perspective \cite{chevalley1948} or for infinite-dimensional Banach Lie groups \cite{neeb2002}. In section \ref{section3}, we present the algorithm mentioned above, which determines the enlarged group $G_{enl}$ from the group $G$, given some algebraic and topological conditions. The algorithm is concretely applied to three physically interesting cases: the Galilei, Poincaré, and Translation groups. The latter is closely related to the Heisenberg group and the Stone-von-Neumann uniqueness theorem \cite{vonneumann1931,mackey1949,rosenberg}, which guarantees the equivalence between wave- and matrix mechanics. We conclude the review in section \ref{section4} by summarizing our results and providing an outlook for further research possibilities. To guarantee the paper's sequential readability, we leave for  Appendix \ref{appendixA} an exposition of the main lines upon which Bargmann's theory was generalized to encompass phase-time dependence in projective representations. 

\subsection{Mathematical framework of quantum mechanics}
In this subsection, we set up the precise mathematical framework that we will use to bring the physicist and mathematician readers to the same page.\footnote{{In order to guarantee sequential readability, we provide here a pinpointed list of mathematical symbols and their meaning: \\ $\mathcal{H}$: Hilbert space; \\ $\mathcal{D}_{A}$: Domain of the operator $A$ $(\mathcal{D}_{A}\subset \mathcal{H})$; \\ $\widehat{\mathcal{H}}$: Projective Hilbert space; \\$\operatorname{Aut}(\cdot)$: Automorphism group of $\cdot$;\\ $PU(\cdot)$: Projective unitary group of $\cdot$;\\ $\pi_1(\cdot)$: First homotopy group of $\cdot$;\\$H^2(\cdot,\circ)$: Second cohomology group of $\cdot$ with coefficients in $\circ$;\\ $\rtimes$: Semi-direct product.}}

{First, we remark that even though the wave function (or a vector in the Hilbert space) is usually considered a state, this definition is not entirely satisfactory. In experiments, we can only measure the probability densities, namely $|\psi|^2$. Therefore, a total global phase cannot be measured, and we must consider elements of the Projective Hilbert space. In other words, we have to describe the states of a quantum system by trace-class linear maps.}
\begin{Definition}
    To every quantum system, there is an associated complex, separable Hilbert space $(\mathcal{H},\langle \cdot, \cdot \rangle)$. Moreover:
    \begin{enumerate}
        \item[$(i)$] The \textbf{states} of the quantum system are all trace-class linear maps
    \begin{equation*}
        \rho:\mathcal{H} \to \mathcal{H}, \; \; \text{for which} \; \; \operatorname{tr}(\rho)=1.
    \end{equation*}
    Furthermore, a state $\rho$ is called \textbf{pure} if
    \begin{equation*}
        \exists \psi \in \mathcal{H}: \forall \Psi \in \mathcal{H}: \rho(\Psi)=\frac{\langle \psi, \Psi \rangle}{\langle \psi,\psi \rangle} \psi.
    \end{equation*}
    If a state is not pure, it is called \textbf{mixed}.
    \item[$(ii)$] The \textbf{observables} of the quantum system are all densely defined self-adjoint operators
    \begin{equation*}
        A: \mathcal{D}_{A} \to \mathcal{H}.
    \end{equation*}
    \end{enumerate}
\end{Definition}
Since in Quantum Mechanics we can only measure transition probabilities, which are not affected by an overall \textit{global} phase in the wavefunction, states are rays in the projective space $\widehat{\mathcal{H}}:= \left(\mathcal{H} \setminus \{ 0 \}\right) / \sim$\footnote{The equivalence relation $x \sim y \iff \exists \theta \in \mathbb{R} : x = e^{i \theta} y$ is used in defining the projective Hilbert space. {Hence, the projective Hilbert space is defined as the Hilbert space set minus the zero element, quotient out by the equivalence relation of being related by a phase.}}.

\begin{Definition}
The function 
$$
\begin{array}{llll}
( \cdot, \cdot): & \widehat{\mathcal{H}} \times \widehat{\mathcal{H}} &\to &\mathbb{R},\\
& \left(\widehat{\varphi}, \widehat{\psi} \right) &\mapsto &\left(\widehat{\varphi},\widehat{\psi}\right) = \frac{|\langle \varphi, \psi \rangle|^2}{|| \varphi ||^2 ||\psi||^2}
\end{array}
$$
is called the \textbf{transition probability}.
\end{Definition}

\begin{Proposition}
    Let $(\mathcal{H},\langle \cdot, \cdot \rangle)$ be a complex Hilbert space and let $T$ denote the set of all trace-class, positive linear operators with trace equal to one that satisfy
    \begin{equation*}
        \exists \psi \in \mathcal{H}: \forall \alpha \in \mathcal{H}: \rho(\alpha)=\frac{\langle \psi, \alpha \rangle}{\langle \psi,\psi \rangle} \psi.
    \end{equation*}
    In other words, $T$ denotes the set of pure states:
    \begin{equation*}
        T=\{\rho:\mathcal{H} \to \mathcal{H}| \rho \; \; \text{is a pure state} \}.
    \end{equation*}
    Then, there exists a bijection between $T$ and $\widehat{\mathcal{H}}$.
\end{Proposition}
\begin{proof}
    Let $\widehat{\psi} \in \widehat{\mathcal{H}}$ be fixed and define the operator associated to $\widehat{\psi}$ as follows
    \begin{equation*}
        \rho_{\widehat{\psi}}: \mathcal{H} \to \mathcal{H}, \; \; \alpha \mapsto \rho_{\widehat{\psi}}(\alpha):=\frac{\langle \psi, \alpha \rangle}{\langle \psi, \psi \rangle} \psi.
    \end{equation*}
    The thus defined operator is well-defined, since if $\widehat{\psi}=\widehat{\phi}$, then $\psi= \mu \cdot \phi$ and we obtain
    \begin{equation*}
        \rho_{\widehat{\psi}}(\alpha)=\frac{ \langle \psi, \alpha \rangle}{\langle \psi, \psi \rangle} \psi=\frac{ \langle \mu \cdot \phi, \alpha \rangle}{\langle \mu \cdot \phi, \mu \cdot \phi \rangle} \mu \cdot \phi=\frac{ \mu^{*} \langle \phi,\alpha \rangle}{\mu^{*} \mu \langle \phi, \phi \rangle} \mu \cdot \phi=\frac{ \langle \phi, \alpha \rangle}{\langle \phi, \phi \rangle} \phi=\rho_{\widehat{\phi}}(\alpha).
    \end{equation*}
    Moreover, the trace of the operator $\rho_{\widehat{\psi}}$ can be computed as
    \begin{equation*}
        \text{Tr}(\rho_{\widehat{\psi}})=\sum_{j=1}^{\text{dim}(\mathcal{H})} \langle e_j, \rho_{\widehat{\psi}}(e_j) \rangle=\sum_{j=1}^{\text{dim}(\mathcal{H})} \frac{ \langle \psi, e_j \rangle}{\langle \psi, \psi \rangle} \langle e_j, \psi \rangle=\sum_{j=1}^{\text{dim}(\mathcal{H})} \frac{|\langle e_j, \psi \rangle|^2}{\langle \psi, \psi \rangle}=\frac{\langle \psi, \psi \rangle}{\langle \psi, \psi \rangle}=1.
    \end{equation*}
    In the last step, we used the Parseval identity. The operator $\rho_{\widehat{\psi}}$ is also positive, as can be seen from
    \begin{equation*}
        \langle \alpha, \rho_{\widehat{\psi}}(\alpha)\rangle= \frac{ \langle \psi, \alpha \rangle}{ \langle \psi, \psi \rangle} \langle \alpha, \psi \rangle=\frac{| \langle \psi, \alpha \rangle|^2}{\langle \psi, \psi \rangle} \geq 0.
    \end{equation*}
    The above discussion provides a well-defined map
    \begin{equation*}
        A: \widehat{\mathcal{H}} \to T, \; \; \widehat{\psi} \mapsto A\left(\widehat{\psi}\right):=\rho_{\widehat{\psi}}.
    \end{equation*}
    We will now prove that $A$ is a bijection in two steps:
    \begin{enumerate}
        \item $A$ is injective: let $\phi, \psi \in \mathcal{H}$ such that $A\left( \widehat{\phi} \right)=A\left(\widehat{\psi}\right)$. As we have shown that $A$ is well-defined, this implies:
        \begin{equation*}
            \forall \alpha \in \mathcal{H}: \rho_{\widehat{\psi}}(\alpha)=\rho_{\widehat{\phi}}(\alpha).
        \end{equation*}
        Since the above equality holds for all $\alpha$, it holds in particular for $\phi$. By explicitly evaluating, we get
        \begin{equation*}
            \frac{\langle \psi, \phi \rangle}{\langle \psi, \psi \rangle} \psi = \frac{\langle \phi, \phi \rangle}{\langle \phi, \phi \rangle} \phi=\phi.
        \end{equation*}
        Letting $\mu=\frac{\langle \psi, \phi \rangle}{\langle \psi, \psi \rangle}$, it follows that $\{\psi,\phi\}$ is a linearly dependent set and this implies that $\widehat{\psi}=\widehat{\phi}$.
        \item $A$ is surjective: let $\rho \in T$ be given. By the definition of $T$, we have
        \begin{equation*}
            \exists \psi \in \mathcal{H}, \forall \alpha \in \mathcal{H}: \rho(\alpha)=\frac{\langle \psi, \alpha \rangle}{\langle \psi ,\psi \rangle} \psi.
        \end{equation*}
        Now we must find $\widehat{\phi} \in \widehat{\mathcal{H}}$, such that $A\left( \widehat{\phi} \right)=\rho$. We claim that $\phi:=\frac{\psi}{\langle \psi, \psi \rangle}$ satisfies this condition. Indeed:
        \begin{equation*}
        \begin{aligned}
            \left(A\left(\widehat{\phi}\right)\right)(\alpha)&= \left(A\left( \widehat{\frac{\psi}{\langle \psi, \psi \rangle}} \right) \right)( \alpha)=\rho_{ \widehat{\frac{\psi}{\langle \psi, \psi \rangle}}} (\alpha)\\
            &=\frac{ \langle \frac{\psi}{\langle \psi, \psi \rangle}, \alpha \rangle}{\langle \frac{ \psi}{ \langle \psi, \psi \rangle}, \frac{\psi}{\langle \psi, \psi \rangle} \rangle}\frac{\psi}{\langle \psi, \psi \rangle}
            =\frac{\frac{1}{\langle \psi, \psi \rangle} \langle \psi, \alpha \rangle}{ \frac{1}{\langle \psi, \psi \rangle} \frac{1}{\langle \psi, \psi \rangle} \langle \psi, \psi \rangle} \frac{\psi}{\langle \psi, \psi \rangle} \\
            &=\frac{ \frac{1}{\langle \psi, \psi \rangle} \langle \psi, \alpha \rangle }{\frac{1}{\langle \psi, \psi \rangle}} \frac{\psi}{\langle \psi, \psi \rangle}=\frac{\langle \psi, \alpha \rangle }{\langle \psi, \psi \rangle} \psi = \rho(\alpha).
            \end{aligned}
        \end{equation*}
    \end{enumerate}
\end{proof}
\begin{Remark}
     In the physics literature, the states are often associated with normalized elements of the Hilbert space. The following proposition gives the precise relation between states and normalized elements of $\mathcal{H}$.
\end{Remark}
\begin{Proposition}
    Let $(\mathcal{H},\langle \cdot, \cdot \rangle)$ be a complex separable Hilbert space and denote with $P'(\mathcal{H})$ the set of normalized elements of $\mathcal{H}$, which are identified through $x \sim y \iff \exists \lambda \in \mathbb{C}^{*}: x= \lambda y$, i.e.
    \begin{equation*}
        P'(\mathcal{H}):=\{\psi \in \mathcal{H}| \langle \psi, \psi \rangle=1 \}/ \sim.
    \end{equation*}
    Then, there exists a bijection
    \begin{equation*}
        B:P'(\mathcal{H}) \to \widehat{\mathcal{H}}.
    \end{equation*}
\end{Proposition}
\begin{proof}
    The proof is constructive. We construct the map $B$ as
    \begin{equation*}
        B: P'(\mathcal{H}) \to \widehat{\mathcal{H}}, \; \; [x]_1 \mapsto  B([x]_1):=\hat{x},
        \end{equation*}
        where $[x]_1$ is the equivalence class of a unit vector $x$ in $P'(\mathcal{H})$, while $\hat{x}$ is its equivalence class in $\widehat{\mathcal{H}}$. The map $B$ is well-defined: suppose $[x_1]_1=[x_2]_1$. This means, that there exists $\lambda \in \mathcal{C}^{*}, |\lambda|=1$, such that $x_1=\lambda \cdot x_2$. In this case, we have:
        \begin{equation*}
            B([x_1]_1)=\widehat{x_1}=\widehat{\lambda \cdot x_2}=\widehat{x_2}=B([x_2]_1).
        \end{equation*}
        We now show that $B$ is a bijection in two steps:
        \begin{enumerate}
            \item $B$ is injective: suppose $x_1, x_2$ are unit vectors, i.e. elements of $P'(\mathcal{H})$ such that $B([x_1]_1)=B([x_2]_1)$. By the definition of $B$ this gives
            \begin{equation*}
                \widehat{x_1}=\widehat{x_2}.
            \end{equation*}Then, by the definition of equivalence in $\widehat{\mathcal{H}}$, there exists a nonzero complex number $\lambda \in \mathbb{C}^{*}$ such that $x_1=\lambda x_2$. As $x_1,x_2$ are unit vectors, this can be satisfied if and only if $|\lambda|=1$, which implies that $x_1,x_2$ satisfy the equivalence relation on $P'(\mathcal{H})$ and thus $[x_1]_1=[x_2]_1$.
        \item $B$ is surjective: Let $\widehat{x} \in \widehat{\mathcal{H}}$. We want to find a unit vector $y$, such that $B([y]_1)=\widehat{x}$. Note that the zero vector is not an element of $\widehat{\mathcal{H}}$ by the very definition of the projective Hilbert space. Finally, define $y:=\frac{x}{\langle x, x \rangle}$. We claim that this does the job. Indeed, it can be readily verified that 
        \begin{equation*}
            B([y]_1)=\widehat{\frac{x}{\langle x,x \rangle}}=\widehat{x}.
        \end{equation*}
        \end{enumerate}
\end{proof}
From now on, we will consider states as rays, i.e., as elements in $\widehat{\mathcal{H}}$. A transformation of a physical system is a transformation of states, hence mathematically a transformation, not of $\mathcal{H}$, but of $\widehat{\mathcal{H}}.$ Not all transformations are symmetries, but only those that preserve the transition probabilities. This leads to the following natural definition due to Wigner \footnote{{Notice that, as defined in Definition 3, the transformations physically represent active transformations since they relate to different frames of reference.}}.
\begin{Definition}
    A bijective map $S:\widehat{\mathcal{H}} \to \widehat{\mathcal{H}}$ with the property
    \begin{equation*}
        \frac{| \langle  S \psi, S \phi \rangle|^2}{||S \psi||^2 ||S \phi||^2}=\frac{|\langle \psi, \phi \rangle|^2 }{||\psi||^2 ||\phi||^2}
    \end{equation*}
    is called a \textbf{symmetry} or \textbf{projective transformation} of the quantum mechanical system $(\mathcal{H}, \langle \cdot,\cdot \rangle)$. The \textbf{automorphism group} $\operatorname{Aut}\left( \widehat{\mathcal{H}} \right)$ of projective transformations is defined as
    \begin{equation*}
        \operatorname{Aut}\left(\widehat{\mathcal{H}} \right):=\{S:\widehat{\mathcal{H}} \to \widehat{\mathcal{H}}|S \; \; \text{is a symmetry} \}.
    \end{equation*}
\end{Definition}
We now present the theorem of Wigner, which relates symmetries to unitary and anti-unitary operators on $\mathcal{H}$.
\begin{Theorem}[Wigner]
Let the set of unitary or anti-unitary operators on a Hilbert space $\mathcal{H}$ be denoted as
\begin{equation*}
    \widetilde{U}(\mathcal{H}):=\{U: H \to \mathcal{H}|U \; \; \text{is unitary or anti-unitary} \}.
\end{equation*}
    Then the map
    \begin{equation*}
        \widehat{\Pi}:\widetilde{U} \to \operatorname{Aut}\left(\widehat{\mathcal{H}}\right), \; \; (\widehat{\Pi}(U))\left(\widehat{\psi} \right):=\widehat{U(\psi)}, \; \; \psi \in \mathcal{H}
\end{equation*}
is a surjective homomorphism with $\ker \left( \widehat{\Pi} \right)=U(1)$.
\end{Theorem}
The abstract formulation of Wigner's theorem leads to two physically interesting corollaries:
\begin{enumerate}
    \item Every symmetry $S \in \text{Aut}\left(\widehat{\mathcal{H}}\right)$ comes from either a unitary or an anti-unitary operator on $\mathcal{H}$ - this follows from the surjectivity of $\widehat{\Pi}$;
    \item If two unitary transformations $U_1,U_2$ give rise to the same projective transformation $\widehat{\Pi}(U_1)=\widehat{\Pi}(U_2)$, then they differ by a phase, i.e. $U_1=e^{it}U_2$ for some $t \in \mathbb{R}$ - this follows from the kernel of $\widehat{\Pi}$ being $U(1)$.
\end{enumerate}
\begin{Remark}
    It is important to mention that several generalizations of Wigner's theorem are already present in the literature. For example, a non-bijective version is presented in \cite{geher2014}. Although not a generalization, the main line of the proof presented in \cite{mouchet2013} could be carried over also to non-separable Hilbert spaces. A generalization to Grassmann spaces can be found in \cite{geher2017}. For further generalizations and a detailed review, consult \cite{chevalier2007}.
\end{Remark}

Note that Wigner's theorem is a result for individual elements of $\operatorname{Aut}\left(\widehat{\mathcal{H}}\right)$. The remainder of this article is devoted to the study of not individual symmetries but symmetry groups. As we study connected Lie groups $G$, only considering symmetries arising from unitary operators will suffice \cite{silva2021}. They are a subgroup of all the symmetries $\text{Aut}\left(\widehat{\mathcal{H}} \right)$.
\begin{Definition}
    Let $(\mathcal{H},\langle \cdot, \cdot \rangle)$ be a complex separable Hilbert space, $\widehat{\mathcal{H}}$ be the projective Hilbert space and
    \begin{equation*}
        \widehat{\pi}:U(\mathcal{H}) \to \text{Aut}\left(\widehat{\mathcal{H}}\right).
    \end{equation*}
The group
    \begin{equation*}
        PU(\mathcal{H}):=\{\widehat{\pi}(U)|U \in {U}(\mathcal{H})\}
    \end{equation*}
    is called the \textbf{projective unitary group} of $\mathcal{H}$.
\end{Definition}
\begin{Remark}
    \label{eq: injective_proj_rep}
    Note that the map $\widehat{\pi}$ is exactly the map $\widehat{\Pi}|_{U(\mathcal{H})}$, that is, it is the map defined in Wigner's theorem restricted to the unitary operators. Also, by constructing equivalence classes of unitary operators, one can show that $PU (\mathcal{H}) = U (\mathcal{H}) / U(1)$.
\end{Remark}

As we want to consider symmetry groups $G$, not just individual symmetries, we introduce the notion of $G$-symmetry.
\begin{Definition}
    Let $(\mathcal{H},\langle \cdot, \cdot \rangle)$ be the Hilbert space associated with a quantum system and let $G$ be a connected Lie group. A $\mathbf{G}$\textbf{-symmetry of a quantum system} or a \textbf{projective unitary representation of} $\mathbf{G}$ is a continuous homomorphism
    \begin{equation*}
        \tau:G \to PU(\mathcal{H}),
    \end{equation*}
    where by continuity with mean with respect to the strong operator topology.
\end{Definition}
Similarly, we introduce the notion of unitary representation of a connected Lie group $G$.
\begin{Definition}
    A unitary representation of a connected Lie group $G$ on a Hilbert space $\mathcal{H}$ is a continuous group homomorphism
    \begin{equation*}
        \mathcal{R}:G \to U(\mathcal{H}),
    \end{equation*}
    where by continuity with mean with respect to the strong operator topology.
\end{Definition}
We shall end this section with a relevant observation concerning $U(\mathcal{H})$. It is stated, more often than never, in the literature that $U(\mathcal{H})$ is not a topological group with respect to the strong operator topology \cite{simms1968, atiyah2004}. However, a rigorous proof of this statement is not presented. By checking the continuity of the group operations explicitly, Schottenloher \cite{schottenloher} proves the contrary, namely that $U(\mathcal{H})$ is a topological group. For other topological aspects of $U(\mathcal{H})$, see \cite{simms1970}.

\section{The lifting problem and its cohomological ramifications}\label{section2}
This section explains how to lift projective unitary representations of a connected Lie group $G$ to unitary representations of $G$. First, we define a lift precisely in the context analyzed here.
\begin{Definition}
    Given a projective unitary representation $\tau:G \to PU(\mathcal{H})$ of a connected Lie group $G$, a unitary representation $\mathcal{R}:G \to U(\mathcal{H})$ is called a lift of $\tau$ if the following diagram commutes:
 \[\begin{tikzcd}
	G && {U(\mathcal{H})} \\
	\\
	&& {PU(\mathcal{H}),}
	\arrow["\mathcal{R}", from=1-1, to=1-3]
	\arrow["\tau"', from=1-1, to=3-3]
	\arrow["\widehat{\pi}", from=1-3, to=3-3]
\end{tikzcd}\]
which means that\footnote{Note that $\widehat{\pi}$ is trivially continuous since it defines the final topology on $PU (\mathcal{H})$.} $\tau= \widehat{\pi} \circ \mathcal{R}$.
\end{Definition}
\begin{Remark}
    The above diagram can be completed by the short exact sequence $\{e\} \to U(1) \to U(\mathcal{H}) \to PU(\mathcal{H}) \to \{e\}$ to obtain the following commutative diagram:
 \[\begin{tikzcd}
	&&& G \\
	{\{e\}} & {U(1)} & {U(\mathcal{H})} & {PU(\mathcal{H})} & {\{e\}.}
	\arrow[from=2-1, to=2-2]
	\arrow[from=2-2, to=2-3]
	\arrow["\widehat{\pi}"', from=2-3, to=2-4]
	\arrow[from=2-4, to=2-5]
	\arrow["\mathcal{R}"', from=1-4, to=2-3]
	\arrow["\tau", from=1-4, to=2-4]
\end{tikzcd}\]
\end{Remark}
In physics, we would like to reduce the theory of unitary projective representations to the theory of unitary representations, as in most cases we work with the Hilbert space $\mathcal{H}$ in practice, not with $\widehat{\mathcal{H}}$. Besides, projective representations may be challenging to handle. The problem with such a reduction is that such a lift does not generally exist for a connected Lie group $G$. There are two obstructions:
\begin{enumerate}
    \item[$(i)$] topological: when the group $G$ is not simply connected;
    \item[$(ii)$] algebraic: when the Lie algebra $\mathfrak{g}$ of $G$ admits central extensions.
\end{enumerate}
However, the projective unitary representations of $G$ are in one-to-one correspondence with unitary representations of an enlarged group $G_{enl}$.
What the enlarged group actually is depends on the obstruction:
\begin{enumerate}
    \item[$(i)$] if there is neither topological nor algebraic obstruction, i.e.
    \begin{equation*}
        \pi_1(G)=\{e\} \; \; \text{and} \; \; H^2(\mathfrak{g},\mathbb{R})=\{e\},
    \end{equation*}
    then the enlarged group is identical to the group itself
    \begin{equation*}
        G_{\text{enl}}=G,
    \end{equation*}
    which means that every projective unitary representation of $G$ can be lifted to a unitary representation of $G$;
    \item[$(ii)$] if there is topological obstruction but no algebraic obstruction, i.e.
    \begin{equation*}
        \pi_{1}(G) \neq \{e\} \; \; \text{and} \; \; H^2(\mathfrak{g},\mathbb{R})=\{e\},
    \end{equation*}then the enlarged group is the universal covering group of $G$
    \begin{equation*}
        G_{\text{enl}}=\widetilde{G};
    \end{equation*}
    \item[$(iii)$] if there is no topological obstruction, but there is algebraic obstruction, i.e.
    \begin{equation*}
        \pi_1(G) =\{e\} \; \; \text{and} \; \; H^2(\mathfrak{g},\mathbb{R}) \neq \{e\},
    \end{equation*}
    then the enlarged group is a central extension of the group $G$;
    \item[$(iv)$] if there is both topological and algebraic obstruction, i.e.
    \begin{equation*}
          \pi_1(G) \neq \{e\} \; \; \text{and} \; \; H^2(\mathfrak{g},\mathbb{R}) \neq \{e\},
    \end{equation*}
    then the enlarged group is a central extension of the universal covering group $\widetilde{G}$ of $G$.
\end{enumerate}
The enlarged group $G_{enl}$ appearing in cases $(ii)-(iv)$ is a central extension of $G$ in either case, but not in the sense of the algebraic theory of groups, but rather in Lie theoretic sense. The cases $(ii)$ and $(iii)$ are relatively simple to handle: $(ii)$ can be seen as a $\pi_1(G)$, whereas case $(iii)$ can be seen as a $U(1)$-extension of $G$. The case $(iv)$ is more subtle and will be discussed later.
\begin{Remark}
    Note that in the case of algebraic obstructions, we have a non-trivial Lie algebra cohomology, which is equivalent to a central extension of the Lie algebra $\mathfrak{g}$ of $G$ by $\mathbb{R}^{n}$ for some $n \in \mathbb{N}$. While it is clear that a Lie group extension of $G$ by $U(1)$ gives rise to a central extension of $\mathfrak{g}$ by $\mathbb{R}$, the converse in general is not true. This is why topological obstructions appear, as we will clarify this in the following.
\end{Remark}

To the best of our knowledge, the literature on Lie group cohomology and central extensions suffers from a somewhat laconism, which motivates a detailed discussion of Lie group extensions.
\begin{Definition}
    Let $G$ be a connected Lie group and $A$ be a connected abelian Lie group, a closed subgroup of $H$. A \textbf{central extension} of $G$ by $A$ is a short exact sequence of Lie groups
  \[\begin{tikzcd}
	{\{e\}} & A & H & G & {\{e\}}
	\arrow[from=1-1, to=1-2]
	\arrow["i", from=1-2, to=1-3]
	\arrow["p", from=1-3, to=1-4]
	\arrow[from=1-4, to=1-5],
\end{tikzcd}\] i.e. the map $i$ is an injective Lie group homomorphism and the map $p$ is a surjective Lie group homomorphism,
such that the image of $i$ is contained in the center of $H$:
\begin{equation*}
    \operatorname{im}(i) \subseteq Z(H), \; \; \text{where} \; \; Z(H)=\{z \in G| \forall g \in G: zg=gz\}.
\end{equation*}
Two central extensions $H_1$ and $H_2$ of $G$ by $A$ are called \textbf{equivalent} if there exists a Lie group isomorphism $\phi:H_1  \to H_2$ such that the following diagram commutes:
\[\begin{tikzcd}
	&& {H_1} \\
	{\{e\}} & A && G & {\{e\}} \\
	&& {H_2}
	\arrow[from=2-1, to=2-2]
	\arrow["{i_1}", from=2-2, to=1-3]
	\arrow["{i_2}"', from=2-2, to=3-3]
	\arrow["\phi", from=1-3, to=3-3]
	\arrow["{p_1}", from=1-3, to=2-4]
	\arrow["{p_2}"', from=3-3, to=2-4]
	\arrow[from=2-4, to=2-5]
\end{tikzcd}\]
\end{Definition}
Clearly, being equivalent as central extensions defines an equivalence relation on the set of all central extensions. We denote by $\operatorname{Ext}_{Lie}(G,A)$ the set of all inequivalent central extensions of $G$ by $A$. It is well known that if we forget the Lie group structure, the central extensions are classified by the second cohomology group:
\begin{equation*}
    Ext(G,A) \cong H^2_{gr}(G,A).
\end{equation*}

The natural extension of the above to Lie groups would be to consider only those cocycles in $H^2_{gr}(G,A)$, which are smooth \footnote{{By smooth, we mean that the cocycles are maps that are infinitely differentiable with respect to the smooth structures of the Lie group $G$ and the Lie group $A$.}}:
\begin{equation*}
    H^2_{s,gr}(G,A):=\{\omega: G \times G \to A| \omega \; \; \text{is a cocycle and it is smooth}\}/ \sim.
\end{equation*}

We will now show that this naive approach does not classify all Lie group extensions.
\begin{Lemma}\label{extensionbundle}
   Let $G,H$ be connected finite-dimensional Lie groups and $A$ be a connected finite-dimensional abelian Lie group. Then the following are equivalent:
   \begin{enumerate}
       \item [$(i)$]  \[\begin{tikzcd}
	{\{e\}} & A & H & G & {\{e\}}
	\arrow[from=1-1, to=1-2]
	\arrow["i", from=1-2, to=1-3]
	\arrow["p", from=1-3, to=1-4]
	\arrow[from=1-4, to=1-5]
\end{tikzcd}\] is a central extension of $G$ by $A$;
       \item [$(ii)$]  $H \to G$ is a principal $A$-bundle, such that $i(A) \subseteq Z(H)$.
   \end{enumerate}
\end{Lemma}
\begin{proof}
We prove this by double implication:
    \begin{enumerate}
        \item $"\Rightarrow"$: Suppose the sequence given in the Lemma is a central extension of $G$ by $A$: as $A$ is closed, by Cartan's closed subgroup theorem it follows that $i$ is a Lie group embedding. By the homogeneous space theorem, $H$ becomes a principal $A$-bundle.
        \item $"\Leftarrow"$: Suppose $H \to G$ is a principal $A$-bundle, such that $i(A) \subseteq Z(H)$. By definition, $H \to G$ is a surjection, and as the local trivializations restrict to diffeomorphisms on fibers, it follows that $A \to E$ is a Lie group injection. By assumption $i(A) \subseteq Z(H)$, which shows that the sequence
        \[\begin{tikzcd}
	{\{e\}} & A & H & G & {\{e\}}
	\arrow[from=1-1, to=1-2]
	\arrow["i", from=1-2, to=1-3]
	\arrow["p", from=1-3, to=1-4]
	\arrow[from=1-4, to=1-5]
\end{tikzcd}\]
is a central extension of $G$ by $A$.
    \end{enumerate}
    \end{proof}
\begin{Corollary}
    $H^2_{s,gr}(G,A)$ does not classify all central extensions of $G$ by $A$ in the Lie group theoretical sense, but only the ones that admit a global smooth section (i.e., a smooth map $s:G \to H$ such that $p \circ s=\text{id}_{G}$), which are exactly the trivial extensions. A fortiori, $H^2_{s,gr}(G,A) = [0].$
\end{Corollary} 
\begin{proof}
    The proof follows the same steps as in the algebraic case (see \cite{lassueur2021}). We choose a section $s : G \to H$ for the central extensions that is not generally a group homomorphism but is of the form 
    $$
    s(g) \cdot s(h) = f(g, h) \cdot s(g h)
    $$
    for some element $f(g, h) \in A$. This defines a map
    \begin{equation*}
    \begin{array}{llll}
    f: &G \times G &\rightarrow &A \\ 
       & (g, h)    &\mapsto         &f(g, h) \equiv s(g) \cdot s(h) \cdot s(g h)^{-1}
    \end{array}
    \end{equation*}
    that can be shown to be a $2$-cocycle.
    One associates to an equivalence class of central extensions an element  $[f] \in H^2_{s,gr}(G,A)$. Since the group $H^2_{s,gr}(G,A)$ contains only smooth $2$-cocycles, the map $s$ of the central extension must also be smooth. By Lemma \ref{extensionbundle}, it follows that the sections of the central extensions can be seen as sections of the principal $A$-bundle $H \to G$. As we required the map $s$ to be globally smooth, it follows that the principal $A$-bundle is the trivial bundle. Equivalently, this means that the sequence splits, which implies that $[f]$ is cohomologous to $[0]$, so $H^2_{s,gr}(G,A)=[0]$.
    Therefore, $H^2_{s,gr}(G,A) = [0]$ classifies Lie group extensions that admit a global smooth section, i.e., the topologically trivial extensions.
\end{proof}
A quite general class of examples where the central extension does not admit a global section occurs when there is a topological obstruction. This happens for Lie groups $G$ that are not simply connected.
\begin{Example}
     \[\begin{tikzcd}
	{\{e\}} & \pi_1(G) & \widetilde{G} & G & {\{e\}}
	\arrow[from=1-1, to=1-2]
	\arrow["i", from=1-2, to=1-3]
	\arrow["p", from=1-3, to=1-4]
	\arrow[from=1-4, to=1-5]
\end{tikzcd}\]
This is a central extension of $G$ by $\pi_1(G)$, but it does not admit a global smooth section unless $G$ is simply connected, as the sequence does not split. Physically relevant examples include:
\begin{enumerate}
    \item The rotation group: $G=SO(3)$, in which case $\pi_{1}(G)=\{1,-1\}$ and $\widetilde{G}=SU(2)$;
    \item The four-dimensional proper ortochronous Lorentz group: $G=SO^{+}(1,3)$, in which case $\pi_{1}(G)=\{1,-1\}$ and $\widetilde{G}=SL(2,\mathbb{C})$.
    \item The metaplectic group for a symplectic vector space $(V,\omega)$: $G=Sp(V,\omega)$,  similarly $\pi_1(G)=\{-1,1\}$ and $\widetilde{G}=Mp(V,\omega)$. This example is relevant in half-form correction in the context of geometric quantization \cite{tuynman2016,schottenloher2024}.
\end{enumerate}
\end{Example}
The proper notion to classify Lie group extensions of $G$ by $A$ is the cohomology theory based on cocycles that are smooth around the identity:
\begin{equation*}
       H^2_{es,gr}(G,A):=\{\omega: G \times G \to A| \omega \; \; \text{is a cocycle, smooth around} \; \; e \in G\}/ \sim.
\end{equation*}
\begin{Proposition}
    There exists a group isomorphism between $\operatorname{Ext}_{Lie}(G,A)$ and $H^2_{es,gr}(G,A)$.
\end{Proposition}
\begin{proof}
    This is an immediate corollary of Lemma \ref{extensionbundle}. Everything goes the same way as in the algebraic group case; however, now, we are able to define a locally smooth section $G \to H$ always, as $H \to G$ is a principal $A$-bundle, which admits a locally smooth section around $e$.
\end{proof}

\begin{Definition}
Let $\mathfrak{a}$ be an abelian Lie algebra, and $\mathfrak{g}$ another Lie algebra. A \textit{central extension} of $\mathfrak{g}$ by $\mathfrak{a}$ is a short exact sequence of Lie algebra homomorphisms
\[
\begin{tikzcd}
0 \arrow[r] & \mathfrak{a} \arrow[r, "\iota", hook] & \mathfrak{e} \arrow[r, "\pi", two heads] & \mathfrak{g} \arrow[r] & 0
\end{tikzcd}
\]
so that $[\mathfrak{a}, \mathfrak{e}] = 0$.
\end{Definition}

Note that a central extension of a Lie group $G$ by $A$ always induces a central Lie algebra extension of $\mathfrak{g}$ by $\mathfrak{a}$ as:
\[\begin{tikzcd}
	{0} & \mathfrak{a} & \mathfrak{h} &  \mathfrak{g} & {0}
	\arrow[from=1-1, to=1-2]
	\arrow["i", from=1-2, to=1-3]
	\arrow["p", from=1-3, to=1-4]
	\arrow[from=1-4, to=1-5].
\end{tikzcd}\]

The converse, in general, does not hold. A sufficient condition for the converse to hold is for $G$ to be simply connected. In this case, a central extension of $\mathfrak{g}$ gives rise to a central extension of $G$. However, this condition is not necessary: the most general result in the reverse direction is due to Neeb \cite{neeb1996}.

\begin{Theorem}[Neeb]
    Let $G,H$ be connected finite-dimensional Lie groups and $A$ be a finite-dimensional connected Lie group and denote by $\mathfrak{g},\mathfrak{h},\mathfrak{a}$ the corresponding Lie algebras. Consider the bracket on $\mathfrak{h}$ viewed as $\mathfrak{g} \times \mathbb{R}$ to be defined by
    \begin{equation*}
        [(X,t),(X',t')]:=([X,X'],\omega([X,X'])),
    \end{equation*}
    where $\omega \in \wedge^2(\mathfrak{g}^{*})$ is a cocycle defining the Lie algebra extension of $\mathfrak{g}$ by $\mathfrak{a}$. Denote by $\Omega$  the corresponding left-invariant $2$-form on $G$:
    \begin{equation*}
        \Omega(g)(d \lambda_{g}(\{e\})v,d \lambda_{g}(\{e\})w):=\omega(v,w), \; \; v,w \in \mathfrak{g} \cong T_{\{e\}}G, \; \; \lambda_{g}(x):=gx.
    \end{equation*}
    Furthermore, denote by $\rho_{g}: x \mapsto xg$ the right translation and similarly let $X_{r}(g)$ be the right-invariant vector field on $G$ given by $X_{r}(g)=d \rho_{g}(1).X$ for $g \in G$. Then, for a Lie algebra central extension
    \[\begin{tikzcd}
	{0} & \mathfrak{a} & \mathfrak{h} &  \mathfrak{g} & {0}
	\arrow[from=1-1, to=1-2]
	\arrow["i", from=1-2, to=1-3]
	\arrow["p", from=1-3, to=1-4]
	\arrow[from=1-4, to=1-5].
\end{tikzcd}\]
A Lie group central extension
    \[\begin{tikzcd}
	{\{e\}} & A & H &  G & {\{e\}}
	\arrow[from=1-1, to=1-2]
	\arrow["i", from=1-2, to=1-3]
	\arrow["p", from=1-3, to=1-4]
	\arrow[from=1-4, to=1-5].
\end{tikzcd}\]
exists iff for each $X \in \mathfrak{g}$ the $1$-form $i(X_r) \Omega$ on $G$ is exact.
\end{Theorem}
Nevertheless, even if the criteria for Neeb's theorem are not satisfied, Lie algebra extensions give rise to Lie group extensions of the universal cover, which is a consequence of the following theorem \cite{landsman1998,foundations2017,vanderschaaf2017}.
\begin{Theorem}\label{theorem1}
    If $G$ is a connected, simply connected Lie group, then the Lie group cohomology $H^2_{es,gr}(G,U(1))$ is isomorphic to the Lie algebra cohomology $H^2(\mathfrak{g},\mathbb{R})$.
\end{Theorem}
We now have the tools to tackle cases $(i)-(iv)$, which we will formulate as proper mathematical theorems. The first and last cases are the well-known theorems due to Bargmann and Cassinelli, which we will not prove. Still, we provide an alternative proof for $(ii)$ without using the most general result due to Cassinelli $(iv)$, but using theorem \ref{theorem1}.
\begin{Theorem}[Bargmann]
    Let $G$ be a connected Lie group which satisfies
    \begin{equation*}
        \pi_{1}(G)=\{e\} \; \; \text{and} \; \; H^2(\mathfrak{g},\mathbb{R})=\{e\}.
    \end{equation*}
    Then every projective representation $\tau:G \to PU(\mathcal{H})$ has a lift to a unitary representation $\mathcal{R}:G \to U(\mathcal{H})$, that is, for every continuous homomorphism $\tau:G \to PU(\mathcal{H})$ there is a continuous homomorphism $\mathcal{R}: G \to U(\mathcal{H})$ with $\tau= \widehat{\pi} \circ \mathcal{R}$.
\end{Theorem}
\begin{Remark}
    Let us note that historically, Bargmann, in his seminal paper \cite{barg}, did not use the language of  Lie algebra cohomology. The modern formulation of his theorem, together with a simplified proof using that $U(\mathcal{H})$ is a topological group, can be found in a recent writing of Schottenloher \cite{schottenloher2008}. See \cite{simms1971} for a proof without multipliers.
\end{Remark}
\begin{Theorem}\label{poincarethm}
    Let $G$ be a connected Lie group, which satisfies
    \begin{equation}
        \pi_1(G) \neq \{e \} \; \; \text{and} \; \; H^2(\mathfrak{g},\mathbb{R})=\{e\}.
    \end{equation}
    Then to every projective representation $\tau:G \to PU(\mathcal{H})$ there corresponds a unitary representation $\mathcal{R}: \widetilde{G} \to U(\mathcal{H})$, where $\widetilde{G}$ is the universal cover of $G$.
\end{Theorem}
\begin{proof}
    Since $\widetilde{G}$ is the universal cover of $G$, they are locally isomorphic, which implies that
\begin{equation}\label{this1}
H^2\left (\widetilde{\mathfrak{g}},\mathbb{R} \right )=H^2( \mathfrak{g}, \mathbb{R})=\{e\}.
    \end{equation}
As $\widetilde{G}$ is simply connected, by theorem \ref{theorem1}, we obtain that 
\begin{equation}\label{this2}
H^2_{es,gr}\left(\widetilde{G},U(1) \right) \cong H^2 \left( \widetilde{\mathfrak{g}},\mathbb{R}\right).
\end{equation}
Relations \eqref{this1} and \eqref{this2} immediately yield
\begin{equation}
    H^2_{es,gr} \left( \widetilde{G},U(1) \right) \cong \{e\}.
\end{equation}
Thus, the projective representations are all equivalent to the trivial class in $H^2_{es,gr}\left( \widetilde{G},U(1) \right)$, which makes the sequence split, making them ordinary representations.
\end{proof}
\begin{Example}
    A broad class of examples can be obtained from the Whitehead Lemma, which states that every real semisimple Lie algebra has $H^2(\mathfrak{g},\mathbb{R})=\{e\}$. Thus, we can take any not simply connected Lie group whose Lie algebra is semisimple. Physically relevant examples within this class include $SO(n)$ for $n \geq 3$ and the Lorentz group.
\end{Example}
\begin{Remark}
    We would like to point out that the proper ortochronous Poincaré group also falls in the previous class of examples, but not due to the Whitehead Lemma, as the Poincaré algebra is not semisimple.
\end{Remark}
\begin{Theorem}
    Let $G$ be a connected Lie group, which satisfies
    \begin{equation}\label{heisenbergthm}
        \pi_1(G)=\{e\} \; \; \text{and} \; \; H^2(\mathfrak{g},\mathbb{R}) \neq \{e\}.
    \end{equation}
Then to every projective representation $\tau:G \to PU(\mathcal{H})$ there corresponds a unitary representation $\mathcal{R}:G_{enl} \to U(\mathcal{H})$, where $G_{enl}$ is a central extension of $G$ by $\mathbb{R}{^n}$ for some $n \in \mathbb{N}$.
\end{Theorem}
\begin{Example}
    By an argument involving the Chevalley-Eilenberg resolution, one can show that the only abelian Lie groups that fall in this class are $\mathbb{R}^{n}$ for $n \geq 2$. These are quite important in quantum mechanics and closely related to the Heisenberg group.
\end{Example}
Finally, the most general case, when the group is neither simply connected nor has vanishing Lie algebra cohomology, is due to Cassinelli. For a proof, consult \cite{cassinelli2004}.
\begin{Theorem}[Cassinelli]
    Let $G$ be a connected Lie group, which satisfies
    \begin{equation}
        \pi_1(G) \neq \{e \} \; \; \text{and} \; \; H^2(\mathfrak{g},\mathbb{R})\neq \{e\}.
    \end{equation}
    Then to every projective representation $T:G \to PU(\mathcal{H})$ there corresponds a unitary representation $R:G_{enl} \to U(\mathcal{H})$, where $G_{enl}$ is a central extension of the universal cover $\widetilde{G}$ of $G$ by $\mathbb{R}^{n}$ for some $n \in \mathbb{N}$.
\end{Theorem}
\begin{Remark}
     The central extension of $\widetilde{G}$ by $\mathbb{R}^{n}$ appearing in the previous theorem is often called the \textbf{universal central extension} of $G$. Universal central extensions for perfect groups have been extensively studied in \cite{raghunathan1994}.
\end{Remark}
\begin{Example}
    Any not simply connected abelian Lie group of dimension greater than two falls within this class. Tori $G=\mathbb{R}^{n}/ \mathbb{Z}^{n}, n \geq 2$ are such examples. However, a physically more relevant example is the Galilei group, which admits a non-trivial central extension by $\mathbb{R}$.
\end{Example}
\section{Enlargement of symmetry groups in physics: a practitioner's guide}\label{section3}
This section provides a concrete algorithm for finding the enlarged symmetry group of a quantum mechanical system, with a particular focus on the symmetry groups appearing in physics.

This algorithm is essentially based on the previously presented theorems. Its key ingredients are the topology of the Lie group in question, i.e. $\pi_1(G)$ and its second Lie algebra cohomology $H^2(\mathfrak{g},\mathbb{R})$. The procedure of finding the enlarged group $G_{enl}$ based on the data  $\left\{\pi_1(G),H^2(\mathfrak{g},\mathbb{R}) \right\}$ is depicted in Figure \ref{illustration} , where  $G^{\star}$ denotes an $\mathbb{R}^{n}$ extension of $G$, which is formally defined as a pair $\left(G^{\star}:=\mathbb{R}^{n} \times G, \cdot \right)$ with the operation
\begin{equation}
    \left(x,g_1 \right)\cdot (y, g_2):=(x+y+\left( \xi_1(g_1,g_2), \xi_2(g_1,g_2),\dots, \xi_{n}(g_1,g_2) \right), g_1 \circ g_2),
\end{equation}
where $\xi_1,\dots,\xi_n \in Z^2_{es}(G,\mathbb{R})$ are $\mathbb{R}$-valued cocycles, which are smooth around the identity. Similarly, $\left(\widetilde{G}\right)^{\star}$ denotes an $\mathbb{R}^{n}$ extension of the universal cover $\widetilde{G}$ of $G$. Formally, it is a pair $\left( \left(\widetilde{G} \right)^{\star}, \cdot \right)$ with the same operation as before.

\begin{figure}[htbp]
\centering
\includegraphics[width=1\linewidth]{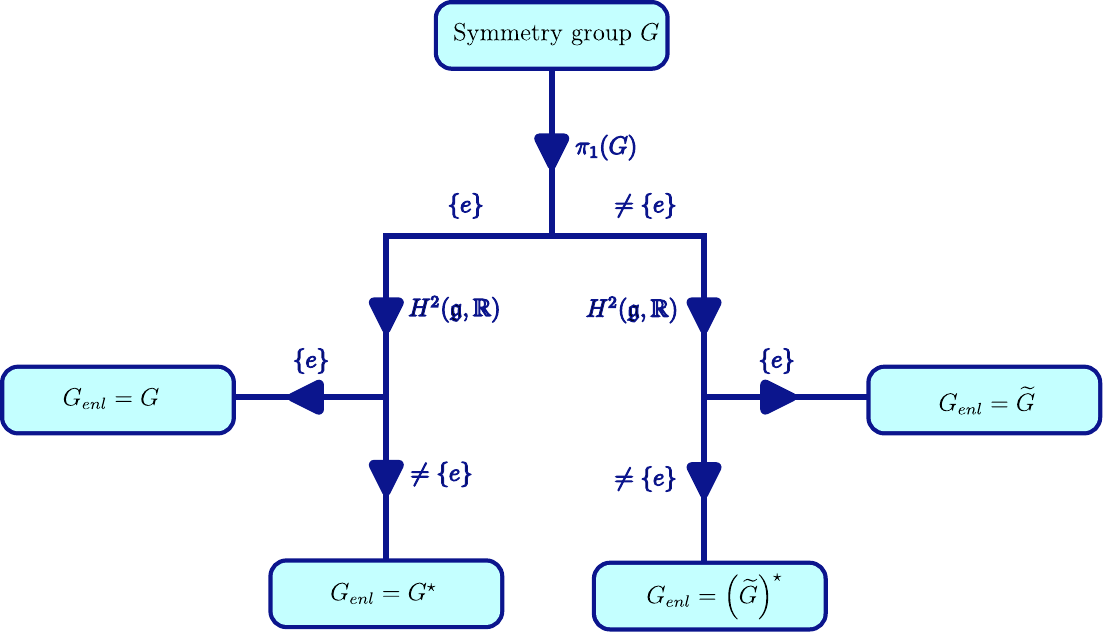}
\caption{Schematic illustration of the procedure that determines the enlarged symmetry group $G_{enl}$ given the data $(G,\pi_1(G),H^2(\mathfrak{g},\mathbb{R} )$.}
\label{illustration}
\end{figure}

This procedure can be readily applied to concrete physical situations at hand. For example, in the case of the rotation group $SO(3)$, we have:
\begin{equation}
\pi_1(SO(3))=\mathbb{Z}_2, \; \; H^2(\mathfrak{so}(3),\mathbb{R})=0.
\end{equation}
It immediately follows, by simply chasing the diagram from Fig. \ref{illustration}, that the projective representations of $SO(3)$ can be lifted to representations of $SU(2)$.

Although this result is well-known, viewing it from this perspective could be illuminating. Moreover, it sheds light on the appearance of spin in quantum mechanics simply from rotational symmetry without the need for any relativistic effect. Note, however, that a relativistic quantum theory is necessary for the correct description of spin-orbit coupling, for instance.

We describe three physically relevant groups for non-relativistic and relativistic quantum mechanics and their projective representations in the following. Essentially, these will be simple applications of the provided algorithm.
\subsection{The Galilei group}
The Galilei group plays a central role in non-relativistic quantum mechanics. In the spirit of Wigner, a non-relativistic particle could be defined as a projective unitary representation of the Galilei group $Gal$. Formally, this group is given by
\begin{equation}
    Gal=\mathbb{R}^4 \rtimes \left( \mathbb{R}^3 \rtimes SO(3) \right).
\end{equation}
It is readily seen that $\pi_{1}(Gal) \neq \{e \}$, since $SO(3)$ appears in the semidirect product. The universal cover of $Gal$ can be straightforwardly constructed as
\begin{equation}
    \widetilde{Gal}=\mathbb{R}^4 \rtimes \left( \mathbb{R}^3 \rtimes SU(2) \right).
\end{equation}
It can be shown that it has a non-trivial Lie algebra cohomology. More precisely we have
\begin{equation}
H^2\left(\mathfrak{Gal},\mathbb{R} \right)=\mathbb{R}.
\end{equation}

Essentially, the non-equivalent central extensions of $\mathfrak{Gal}$ are labeled by a real number, which is physically interpreted as the mass. In Weinberg's terminology, the mass is a central charge in the extended algebra. This interpretation also leads to a well-known superselection rule of particles with different masses, called the \textit{Bargmann superselection rule} \cite{wightman1995,giulini1996,hernandez2012,annigoni2013}. Hence, the Galilei group falls into the most general case, to which Cassinelli's theorem has to be applied. 

\subsection{The Poincaré group}
The (proper ortochronous) Poincaré group is the key object in Wigner's classification of relativistic particles. Although its irreducible representations are well known and present in almost every physics textbook, rigorous accounts and arguments are often missing. To this end, we would like to point out the class to which the (proper ortochronous) Poincaré group belongs in the scheme developed here. Recall that the (proper ortochronous) Poincaré group is a semidirect product of the spacetime translation group and the connected component of the proper ortochronous Lorentz group
\begin{equation}
    Poin =\mathbb{R}^4 \rtimes SO^{+}(1,3).
\end{equation}
Its universal cover is obtained from the double-sheeted $SL(2,\mathbb{C})$ covering as
\begin{equation}
    \widetilde{Poin}=\mathbb{R}^4 \rtimes SL(2,\mathbb{C}),
\end{equation}
from which we can conclude that it is not simply connected. However, its Lie algebra cohomology is vanishing even though the Lie algebra is not semisimple (for a proof, see \cite{vanderschaaf2017}). Hence, the Poincaré group falls in the case of Theorem \ref{poincarethm}. In Weinberg's terminology, all the central charges can be eliminated, as there is no non-trivial central extension, due to the absence of algebraic obstructions.

\subsection{The Heisenberg group}
Let us consider the Hilbert space $\mathcal{H}=L^2\left(\mathbb{R}^{2n} \right)$. The translations in position and momentum space are given by
\begin{equation}
    (P_a \psi)(x)=\psi(x-a), \; \; (M_a \psi)(x)=e^{iax}\psi(x), \; \; \text{respectively}.
\end{equation}
We can easily define a projective unitary representation of $\mathbb{R}^{2n}$ on $\mathcal{H}$ as
\begin{equation}
    \tau:\mathbb{R}^{2n} \to PU(\mathcal{H}), \; \; \tau(x,y)=P_x \circ M_y,
\end{equation}
since the two operators commute up to a phase. However, we cannot lift this projective representation to a unitary representation of $\mathbb{R}^{2n}$, since it has a non-vanishing Lie-algebra cohomology.  However, the projective unitary representation of $\mathbb{R}^{2n}$ is in one-to-one correspondence with the (unique) unitary representation of the Heisenberg group, which is a central extension of $\mathbb{R}^{2n}$ by $\mathbb{R}$, as $\mathbb{R}^{2n}$ falls in the case of theorem \ref{heisenbergthm}. Moreover, this representation is unique, as guaranteed by the Stone-von-Neumann theorem. Thus, the projective representations of $\mathbb{R}^{2n}$ lead to the unique representation of the position and momentum operators, respectively. They implement the commutation relation on (a dense subspace of) $\mathcal{H}$.

\section{Summary and outlook}\label{section4}
This article presented a detailed and mathematically precise review of the representation theory used in quantum mechanics. We noted the differences between ordinary group cohomology and Lie group cohomology, where the topological and smooth aspects have been discussed in depth. An algorithm was provided, which relates the enlarged group $G_{enl}$ of a quantum system to the $G$-symmetries described by a connected Lie group $G$. From our presentation, it becomes clear in which cases one has to consider the universal cover, a central extension, or a central extension of the universal cover, thus making the ideas presented in Weinberg's textbook \cite{weinberg2005} mathematically precise. All of this boils down to two properties of the Lie group $G$ in question: $\pi_1(G)$ and $H^2(\mathfrak{g},\mathbb{R})$, which are topological and algebraic, respectively. The algorithm was applied to three concrete cases of physical relevance. 

In this approach, although spin enters as a quantum label necessary to represent Poincaré symmetries in the Hilbert space, it can also be viewed as a consequence of quantum mechanics' mathematical formalism and consistency. It appears in any quantum system exhibiting rotational symmetry through the projective representations of $SO(3)$, or equivalently, as our theorems show, the unitary representations of $SU(2)$, respectively. We would also like to point out that projective representations are necessary, not just a matter of choice for the Galilei group. It can be shown that the Schrödinger equation is Galilei invariant iff the wave function transforms up to a phase that can not be eliminated.

To conclude this review, we would like to mention that recent progress has been made in the context of the representation theory of Lie supergroups. In particular, Mackey's imprimitivity theorem has been generalized to that setting \cite{carmeli2006}. It would be interesting to see if Bargmann's theorem, or more generally Cassinelli's theorem, can also be generalized, leading to a clear description of how to enlarge symmetry groups in the SUSY setting.

\section*{Acknowledgments}
L.Cs. would like to thank A. Craciun, W. Wawrow, A. Vats for helpful discussions and professors  Z.Néda and V.Chiș for sterning his interest in the topic. J.M.HdS. would like to thank CNPq (grant No. 307641/2022-8) for financial support. The work of L.Cs. was supported by a grant of the Ministry of Research, Innovation and Digitization, CNCS/CCCDI - UEFISCDI, project number PN-IV-P8-8.1-PRE-HE-ORG-2023-0118, within PNCDI IV and Collegium Talentum Hungary.

\appendix
\section{ Generalization of Bargmann's Theory}\label{appendixA}
Here, we shall depict the main steps of the formalism presented in \cite{voa}, leading to a generalization of Bargmann's theory of projective representations including time dependence\footnote{Spacetime dependence may be directly accomplished by straightforward adaptation of the formalism we are about to review.} on the phase. Before doing so, let us briefly recall the direct relation between time-dependent phases and gauge fields appearing even in non-relativistic quantum mechanics. 

Roughly speaking, consider a wave function $\psi$ describing a quantum system whose dynamics are dictated by the Schrödinger equation. Of course, any element of the set $\{e^{i\alpha\psi}\}$ for $\alpha$ constant (and real) keeps the dynamical equation invariant and the physical system description intact. Nevertheless, if $\alpha=\alpha(t,{\bf x})$, then the dynamical equation is not invariant anymore, and to restate invariance, compensating fields are in order. In the simple case at hand, the prescription means new time and spacial derivatives $\partial_t\mapsto \partial_t+i\varphi$ and ${\bf \nabla}\mapsto {\bf \nabla}-i{\bf A}$, respectively with (gauge) fields transforming by $\varphi'=\varphi-\partial_t\alpha(t,{\bf x})$ and ${\bf A}'={\bf A}+{\bf \nabla}\alpha(t,{\bf x})$. In an attempt to interpret and understand gauge freedom within the scope of projective representations, a formalism extending the theory was developed \cite{voa}. From now on, we will focus on the main steps of this generalization.   

The first aspect we shall emphasize is that probabilities can only be summed up for a countable set of events. This simple statement naturally leads us to appreciate the basic formulation for this generalization in the context of analytic Borel structures. Moreover, use is made of a specific fiber bundle whose typical fiber is given by a Hilbert space $\mathcal{H}$. 

Consider a fiber bundle with a base manifold given by $\mathbb{R}$. For every point $t\in\mathbb{R}$, we associate a Hilbert space $\mathcal{H}_t$ performing the fiber. The Hilbert bundle is defined as the set $\{(t,\psi)|t\in\mathbb{R},\psi\in\mathcal{H}_t\}\:=\mathbb{M}$. A Hilbert bundle cross-section is defined by
\begin{eqnarray}
\psi&:&\left.\mathbb{R}\rightarrow \mathcal{H}_t\right.\nonumber\\
&& t\mapsto \psi_t, \forall t\in\mathbb{R}.
\end{eqnarray}
Therefore, for any cross-section $\psi$ and $(t_0,\phi_{0})\in\mathbb{M}$ it is possible to compose the inner product $(\phi_{0},\psi_{t_{0}})$ in such a way that any cross-section define complex-valued function in $\mathbb{M}$. Recall that a Borel set may be generically defined as one formed through the countable union, countable intersection, and complement of open sets. Besides, a given function, say $f:X\rightarrow Y$, is said to be a Borel function if $f^{-1}(O)$ is a Borel set for every open set $O\subset Y$. 

\begin{Definition}
	A \textbf{Borel bundle} is a Hilbert bundle $\mathbb{M}$ endowed with a Borel analytic structure\footnote{A continuous image of a Borel set in a Banach space is called an analytic Borel structure. In particular, it is Lebesgue measureable.} in $\mathbb{M}$ such that
	\begin{itemize}
		\item Let $\pi(t,\psi)=t$ be the fibre bundle projection. Then, $\mathbb{E}\subseteq \mathbb{R}$ is a Borel set iff $\pi^{-1}(\mathbb{E})$ is a Borel set on $\mathbb{M}$.
		\item There exist countable sections $\{\psi^i\}$ ($i=1,2,\cdots$) such that: 
		\begin{itemize}
		\item their corresponding complex-valued functions on $\mathbb{M}$ are Borel functions;
		\item no two different points $(t_k,\phi_k)\in \mathbb{M}$ ($k=1,2$) give the same values for every $\psi^i$, except if $\phi_1=\phi_2=0$;
		\item the mapping $t\rightarrow (\psi^i_t,\psi^j_t)$ is a Borel function for all $i$ for all $j$. 
		\end{itemize}	
	\end{itemize}	
\end{Definition}	
Finally, a section is called a Borel cross-section if the function defined on $\mathbb{M}$ by the section is a Borel function. A couple of definitions suitable for our approach are still necessary.  

\begin{Definition}
	Let $\mu$ be a Lebesgue measure on $\mathbb{R}$. A section $t\rightarrow \psi_t$ is \textbf{square summable with respect to $\mu$} if 
	\begin{equation}
	\int_\mathbb{R}(\psi_t,\psi_t)d\mu(t)<\infty.
	\end{equation} 
\end{Definition}

From these definitions, it is possible to say that the space $\mathcal{L}^2(\mathbb{R},\mu,\mathcal{H})$ containing all classes of square summable sections forms a separable Hilbert space with inner product given by 
\begin{equation}
(\phi,\psi)=\int_\mathbb{R}(\phi_t,\psi_t)d\mu(t).
\end{equation} Also, using this construction, it is conceivable to assign to each time a self-adjoint operator $O_t:\mathcal{H}_t\rightarrow \mathcal{H}_t$ performing, at the formalism level, the measuring action.



A relevant isomorphism between Hilbert bundles, $\mathbb{M}\simeq\mathbb{M}'$, is defined as follows: $U:\mathbb{M}\rightarrow\mathbb{M}'$ is a Borel isomorphism such that for every $t\in\mathbb{R}$
\begin{eqnarray}
U|_{t\times \mathcal{H}_t} \Rightarrow \exists \hspace{.2cm} t'\times \mathcal{H'}_{t'}\in \mathbb{M}',
\end{eqnarray} with unitarity. Of course, for $\mathbb{M}'=\mathbb{M}$, we have an automorphism. 
     
Consider, as usual, a connected Lie group $G$ and denote by $\mathcal{A}$ the set of automorphisms of $\mathbb{M}$. Consider also the action $G\times \mathbb{R}$ of $G$ on $\mathbb{R}$ given by $t\mapsto r^{-1}t$ for all $r\in G$ and $t\in\mathbb{R}$. The mapping $r\rightarrow U_r$ (from $G$ to $\mathcal{A}$) is a representation of $G$ if $U_rU_s=e^{i\xi(r,s,t)}U_{rs}$. Besides, it is considered that: {\it i)} the action of $G$ on $\mathbb{M}$ is transitive and smooth, {\it ii)} $\xi$ is differentiable in $t\in\mathbb{R}$. Denote by $\mathcal{F}$ the set of all real and differentiable functions on $\mathbb{R}$. Let us revisit the ray definition to the relevant quantities for our purposes. 
\begin{Definition}
	\textbf{The ray} of a given cross-section $\psi$ is defined by $\Psi=\{e^{i\zeta(t)}\psi(t), \zeta\in\mathcal{F}\}$ and the \textbf{operator ray} corresponding to a given automorphism $U$ is $\mathbb{U}=\{\tau U| t\mapsto \tau(t)\in \mathcal{F}, |\tau|=1\}$. 
\end{Definition} As usual in the standard formulation, the continuity concept allows for the current physical interpretation that probability transitions vary continuously with variation of elements in $G$. The idea of continuity can be expressed as follows: for any $r\in G$, any ray $\Psi$ and any positive $\varepsilon$, there exists a neighborhood $\mathcal{V}(\ni r) \subset G$ such that 
\begin{eqnarray}
d_t({\bf U_s}\Psi,{\bf U_r}\Psi)<\varepsilon,
\end{eqnarray} if $s\in\mathcal{V}$ and $t\in\mathbb{R}$, where $d_t(\Psi_a,\Psi_b)=\inf||\psi_a-\psi_b||_t$ and $\psi\in\Psi$. With this concept of continuity, it is possible to show that it is always feasible to find continuous representatives\footnote{See Ref. \cite{silva2021}, Theorem 1, for detailed and commented proof. The unique necessary adaptation of the proof presented in \cite{silva2021} to the case at hand is the consideration that $t$ must belong to a compact subset $\mathcal{S}$ of $\mathbb{R}$.} $U_r\in{\bf U}_r$. The associative law together with $U_rU_s=e^{i\xi(r,s,t)}U_{rs}$ may be used to show that continuity of phase factors is inherited from the continuous representatives\footnote{See also \cite{silva2021}, Lemma 1, replacing $\omega$ by $\exp{[i(r,s,t)]}$.}. Additionally, notice that given a continuous representative $U_r$, another one may be reached by $U'_r=e^{i\zeta(r,t)}U_r$ for $\zeta$ real, continuous in $r$. In fact, $||U'_s\psi-U'_r\psi||<\varepsilon$ implies $2-2\text{Re}(e^{i\kappa}U_s\psi,U_r\psi)<\varepsilon^2$, where $\kappa\equiv \zeta(s,t)-\zeta(r,t)$ and we assume $t\in \mathcal{S}\subset\mathbb{R}$. Calling $a_p$ ($b_p$) the real (imaginary) part of $U_p\psi$, the previous relation reads $(1-a_s\cos(\kappa)-b_s\sin(\kappa)-a_r)<\varepsilon^2/2$. It is readily verified that iff the limit $\kappa\rightarrow 0$ is well defined (for $r,s\in\mathcal{V}\subset G$), the new representatives are also continuous.  As it can be verified, the relation $\xi'(r,s,t)=\xi(r,s,t)+\Delta[\zeta]$ is reflexive, transitive, and symmetric, leading to a true equivalence relation. Thus, two exponents $\xi$ and $\xi'$ are equivalent if $\xi'(r,s,t)=\xi(r,s,t)+\Delta[\zeta]$. 

The phase freedom in the ray representation allows for a relevant generalization, encompassing local groups. In fact, take $\theta(t)\in \mathcal{F}$ and select representatives ($e^{i\theta(t)}U_r$ and $e^{i\theta'(t)}U_s$) such that 
\begin{equation}
(e^{i\theta(t)}U_r)(e^{i\theta'(t)}U_s)=e^{i(\theta(t)+\theta'(r^{-1}t)+\xi(r,s,t))}U_{rs},\label{lg}
\end{equation} from which a group structure may be read for elements belonging to $\{\theta(t),r\}$ with product given by $\{\theta(t),r\}\{\theta'(t),s\}=\{\theta(t)+\theta'(r^{-1}t)+\xi(r,s,t),rs\}$, for $r,s\in G$. As usual, call this group $H$. If $e$ is the identity in $G$, then $\bar{e}=\{0,e\}$ is the identity in $H$ and elements given by $\{\theta(t),e\}$ form an Abelian subgroup $T\subset H$. The novelty appearing in this generalization of Bargmann's theory concerning these $H$ (local) groups comes from the acting of $G$ in $\mathbb{R}$: due to this action ($t\mapsto r^{-1}t$) $H$ is to be considered equal to $T\rtimes G$ so that $H$ is a semicentral extension of $G$. However, it is still true that locally $G\simeq H/T$ and if two exponents are equivalent, their respective associated semicentral extension groups are homomorphic \cite{barg}. Finally, the exponents are proved to be continuous, but it is also possible to show that they are additionally differentiable in $r$ and $s$. The standard method to prove this important property is the so-called Iwasawa construction \cite{iwa,barg,voa}. See also \cite{ald} (Theorem 1 and Lemma 1) for a detailed account of exponents differentiability via Iwasawa construction. 

Generally speaking, the above construction already generalizes phase dependence of projective representations to encompass time explicitly. The rest of this appendix deals with framing such a generalization, {\it mutatis mutandis}, into the main text algorithm regarding the enlargement of symmetry. To do so, we shall sketch the main aspects of the ingenious strategy developed in Ref. \cite{voa} to construct the algebra related to the $H$ groups.

Firstly, we call attention to the fact that there is a very useful particularization in the analysis that shall be used from now on. A given local exponent $\xi(r,s,t)$ of $G$ is said to be canonical if it is differentiable in all variables and equals zero if $r$ and $s$ are elements of a one-parameter subgroup. It is possible to show \cite{barg} that locally every exponent is equivalent to a canonical one. 

It is considered an embed of $H$ in a Lie group (with manifold structure performed by a Banach space) so that a standard relation between group and algebra is warranted with convergence of the Baker-Campbell-Hausdorff series. Hence, take $\tilde{H}$ as the closure of $H$. The semisimple extension of $T$ with $G$ is inherited to its closure $\tilde{N}$, that is $\tilde{H}=\tilde{T}\rtimes G$ and, therefore, $\tilde{H}\ni (n,g)$ with $n\in \tilde{N}$ and $g\in G$. Observe that $(n,g)(n',g')=(ngn'g^{-1},gg')$. Thus $gn'g^{-1}\in\tilde{N}$ defines an automorphism, say $\mathcal{A}_g:=gn'g^{-1}\in \text{Aut}(\tilde{N})$. Elements of $H\subset\tilde{H}$ are given by $\{f(t),g\}$ and to give the right account of group action upon $t\in\mathbb{R}$ (as discussed), the restriction of $\mathcal{A}$ to $H$ shall given by $(\mathcal{A}_g f)(t):=f(g^{-1}t)$.  

Consider any two fixed vectors $A$ and $B$ of the Lie algebra, and $\tau A$ and $\tau B$ (for $\tau \in \mathbb{R}^*$), two local one-parameter subgroups. Define the quantity 
\begin{equation}
q(\tau A,\tau B):=\lim_{\tau\rightarrow 0}\tau^{-2}(\tau A)(\tau B)(\tau A)^{-1}(\tau B)^{-1}. 
\end{equation} Using the standard exponential mapping between Lie algebra and Lie group and the Baker-Campbell-Hausdorff series, one gets 
\begin{eqnarray}
e^{\pm\tau A}e^{\pm\tau B}=\exp{\Big(\pm\tau A\pm\tau B+\frac{\tau^2}{2}[A,B]\pm O(\tau^3)\Big)},
\end{eqnarray} up to second order in $\tau$. Therefore, for $\tau$ sufficiently small, we are left with $\exp{(\tau^2[A,B])}$, from which the bilinear, antisymmetric commutator can be fairly identified, i.e., $q=[A,B]$. Now, the computation of $[\bar{a},\bar{b}]$ 
\begin{eqnarray}
[\bar{a},\bar{b}]=[\{f,a\},\{y,b\}]=\lim_{\tau\rightarrow 0}\tau^{-2}(\tau\bar{a}\tau\bar{b})(\tau^{-1}\bar{a}^{-1}\tau^{-1}\bar{b}^{-1}),
\end{eqnarray} leads to 
\begin{eqnarray}
    [\bar{a},\bar{b}]=\{\mathfrak{F}(f,y)+\Xi(a,b,t),[a.b]\},
\end{eqnarray} where 
\begin{eqnarray}
\Xi(a,b,t)&=&\left.\lim_{\tau\rightarrow 0}\tau^{-2}[\xi((\tau a)(\tau b),(\tau a)^{-1}(\tau b)^{-1},t)+\xi(\tau a,\tau b,t)\right.\nonumber\\&+&\left.\xi((\tau a)^{-1},(\tau b)^{-1},(\tau b)^{-1}(\tau a)^{-1}t],\right.\label{qql}
\end{eqnarray}  are called infinitesimal exponents and $\mathfrak{F}(f,y)$ can be bypassed for this general exposition. It is worth stressing that the Jacobi identity is warranted for the bilinear form, and then, the algebra is recovered fully. As a last step, we call attention to the following particularization concerning canonical exponents: if two canonical exponents are equivalent, then $\xi'=\xi+\delta[\Lambda]$, where $\Lambda$ is a linear real form defined in the algebra \cite{barg}. By inserting this particularized equivalence into (\ref{qql}), we arrive at a similar equivalence relation for the infinitesimal exponents, that is, $\Xi'=\Xi+\Delta[\Lambda]$. From now on, there is a {\it pari passu} reasoning between this formulation and Bargmann's theory. With effect, as in the standard case, to every local exponent, there corresponds an infinitesimal exponent\footnote{It is more or less obvious from the outlined construction (see \cite{voa}, Theorem 3).}. Moreover, two canonical exponents are equivalents iff their infinitesimal exponents counterparts are also equivalents. This similarity assigns the generalization at hand to the algebraic branch of the enlarging symmetry algorithm highlighted in the main text. Besides, from the topological aspect, we call attention once again to the fact that, just as in the usual case, a given exponent is equivalent to a canonical one only locally, and the obstructions to the global validity of the analysis rest again upon the group manifold simply connectivity.

\printbibliography

\end{document}